\documentclass{llncs}
\usepackage{llncsdoc}
\usepackage[numbers]{natbib}
\usepackage[usenames]{xcolor}
\usepackage{xcolor}
\usepackage{amsmath}
\usepackage{amsfonts}
\usepackage{threeparttable}
\usepackage{subfigure}

\makeatletter
\renewcommand\bibsection%
{
  \section*{\refname
    \@mkboth{\MakeUppercase{\refname}}{\MakeUppercase{\refname}}}
}
\makeatother

\begin{document}

\title{The Routing of Complex Contagion in Kleinberg's Small-World Networks}
\author{Wei Chen\inst{1} \and Qiang Li\inst{2}
	\and Xiaoming Sun\inst{2} \and Jialin Zhang\inst{2}}

\institute{Microsoft Research
\and
Institute of Computing Technology, Chinese Academy of Sciences}
\maketitle
\begin{abstract}
In Kleinberg's small-world network model, strong ties are modeled as deterministic edges in the
	underlying base grid and weak ties are modeled as random edges connecting remote nodes.
The probability of connecting a node $u$ with node $v$ through a weak tie is proportional to
	$1/|uv|^\alpha$, where $|uv|$ is the grid distance between $u$ and $v$ and $\alpha\ge 0$ is the
	parameter of the model.
Complex contagion refers to the propagation mechanism in a network where each node is activated only
	after $k \ge 2$ neighbors of the node are activated.
	
In this paper, we propose the concept of routing
	of complex contagion (or {\em complex routing}), where 
	at each time step we can select one eligible node (nodes already having two active neighbors)
	to activate, with the
	goal of activating the pre-selected target node in the end. 
We consider decentralized routing scheme where only the links connected to already activated
	nodes are known to the selection strategy.
We study the routing time of complex contagion and compare the result with simple routing and complex diffusion
(the diffusion of complex contagion,
where all eligible nodes are activated immediately
in the same step with the goal of activating all nodes in the end).

We show that for decentralized complex routing, the routing time is lower bounded by a polynomial in $n$ (the number of nodes in the network)
	for all range of $\alpha$ both in expectation and with high probability
	(in particular, $\Omega(n^{\frac{1}{\alpha+2}})$ for $\alpha \le 2$ and
		$\Omega(n^{\frac{\alpha}{2(\alpha+2)}})$ for $\alpha > 2$ in expectation).
%		, while the routing time of simple contagion has polylogarithmic upper bound when $\alpha = 2$.
Our results indicate that
	complex routing is exponentially harder than both simple routing and
	complex diffusion at the sweetspot of $\alpha =2 $.
\end{abstract}


\section{Introduction}

Social networks are known to be the medium for spreading disease, information, ideas, innovations, and other types of behaviors.
Social scientists have been studying social networks and diffusions in the networks for decades,
	and many of the research results are inspirational to researches in the intersection of
	social science, economics, and computation on modeling social networks and diffusions in them.

In the seminal work~\cite{Granovetter1973strength,Granovetter1974getting}, Granovetter classified relationships in a social network
	as strong ties and weak ties.
Strong ties represent close relationships, such as family members and close friends, while weak ties
	represent acquaintance relationship that people casually maintain.
The surprising result in this study is that people often obtain important job referrals leading to their
	current jobs through weak ties instead of strong ties, which leads to the popular term
	{\em the strength of weak ties}.
His research demonstrated the importance of weak ties in information diffusion in social networks.
Another famous experiment related to information diffusion is Milgram's small-world experiment
	\cite{Milgram1967small}, in which Milgram asked subjects to forward a letter to their friends
		in order for the letter to reach a person not known to the initiator of the letter.
The result showed that on average it takes only six hops to connect two people in U.S. unknown to each other,
	hence the famous term of {\em six-degree of separation}.

The above studies motivated the modeling of small-world networks \cite{Watts1998collective,Kleinberg2000small}.
Watts and Strogatz modeled the small-world network as a ring where nodes close to one another in ring
	distance are connected representing strong ties, and some strong ties are rewired to connect
	to other random nodes on the ring, which represent weak ties \cite{Watts1998collective}.
They also proposed short diameter (the distance between any pair of nodes is small) and
	high clustering coefficient (the probability that two friends of a node are also friends of each
	other) as two characteristics of small-world networks.
Kleinberg~\cite{Kleinberg2000small} improved the model of Watts and Strogatz by building a small-world network on top of a base
	grid, where grid edges representing strong ties, and each node $u$ initiating a weak tie
	connecting to another node $v$ with probability proportional to $1/|uv|^{\alpha}$, where
	$|uv|$ is the grid distance between $u$ and $v$ and $\alpha$ is the small-world parameter.
Kleinberg showed that when $\alpha$ equals the dimension of the grid, the decentralized greedy routing,
	where in each routing step the current node routes the message to its neighbor with grid distance
	closest to the target node, achieves efficient routing performance \cite{Kleinberg2000small}.
This efficient decentralized routing behavior qualitatively matches the result of Milgram's small-world
	experiment.
Kleinberg further showed that when $\alpha$ is not equal to the grid dimension, no decentralized routing
	scheme could be efficient, and in particular, the small-world model of Newman and Watts \cite{Newman1999scaling} corresponds
	to the one-dimensional Kleinberg's model with $\alpha = 0$.
Kleinberg's small-world network model is the one we use in this paper.

In another work \cite{Granovetter1978threshold}, Granovetter proposed the threshold model to characterize
	diffusions of rumors, innovations, or riot behaviors.
An individual in a social network is activated by a certain behavior only when the number of her neighbors
	already adopting the behavior exceeds a threshold.
This threshold model motivated the linear threshold, fixed threshold, and general threshold models
	proposed by Kempe et al. \cite{Kempe2003maximizing}, and
	is directly related to the model of complex contagion we use in this paper.

More recently, Centola and Macy \cite{Centola2007complex} classified the threshold model
	into simple contagion and complex contagion.
Simple contagion refers to diffusion models with threshold being one on every node, which means that
	a node can be activated as long as there is one active neighbor.
Simple contagion corresponds to diffusions of virus or simple information, where one can get activated
	by simply receiving the virus and information.
Complex contagion, on the other hand, refers to diffusion models with threshold at least two, meaning that
	a node can be activated only after multiple of its neighbors are activated.
Complex contagion corresponds to diffusions requiring complex decision process by individuals, such as
	adopting a costly new product, adopting a disruptive innovation, etc, where people usually need
	multiple independent sources of confirmation about the utility of the new product or new innovation
	before taking the action.
The important point Centola and Macy argued is that, while weak ties are effective in transmitting
	information quickly across a long range in a network, they may not be as effective in complex contagion.
This is
	because for complex contagions to spread quickly in a network, it requires weak ties forming
	not only
	long bridges connecting different regions of the network but also {\em wide} bridges in the sense that
	many weak ties can work together to bring the contagion from one region of the network to another
	region of the network.

Motivated by the above work, Ghasemiesfeh et al. provided the first analytical study of complex contagion
	in small-world networks \cite{Ghasemiesfeh2013complex}. 
They studied the diffusion of $k$-complex contagion (or {\em $k$-complex diffusion}),
	where all nodes have threshold $k$ and all nodes with at least $k$ active neighbors are activated right away. 
They showed that the {\em diffusion time}, which is the time for the diffusion to
	activate all nodes in a network starting from $k$ initial seed nodes connected
	with strong ties, is polylogarithmic to the size of the network when $\alpha = 2$. 
Ebrahimi~\cite{ebrahimi2015complex} further generalized the results 
	and proved that the diffusion time for $k$-complex diffusion has polylogarithmic upper bound when $\alpha \in (2,\frac{2(k^2+k+1)}{k+1})$ in Kleinberg's grid model.
They also show that in Kleinberg's model with $\alpha$ outside this range, the diffusion time is lower
	bounded by a polynomial in $n$.

In this paper, we go beyond the diffusion of complex contagion (or {\em complex diffusion}), to
	study a new propagation phenomenon closer to decentralized routing in \cite{Kleinberg2000small},
	which we call the {\em routing of complex contagion} (or {\em complex routing}).
In complex routing, we model weak ties as directed edges as in \cite{Kleinberg2000small}, and
	study the time for two seed nodes connected by a strong tie to activate
	a target node $t$ farthest on the grid (we call it the {\em routing time}).
At each step
	only one new node can be activated, and the decision of which node to activate is decentralized which means it is only based on
	the current activated nodes and their outgoing weak tie neighbors as well as the underlying grid, same
	as decentralized routing in \cite{Kleinberg2000small}.
Such decentralized routing behavior corresponds to real-world phenomenon where a group of people
	want to influence a target person by influencing intermediaries between the source group
	and the target person, and influencing these intermediaries requires effort and thus has to be
	carried out one at a time.
Active friending \cite{YangHLC13} is an application similar to the above scenario
	recently proposed in the context of
	online social networks such as Facebook for increasing the chance of a target user accepting
	the friending request from the source.

\subsection{Our results}

%In this paper, we study complex contagion in the two-dimensional Kleinberg's small-world network model
%	with general parameter $\alpha$, and address several issues either left open or not
%	studied in \cite{Ghasemiesfeh2013complex}.
%We discuss below our main results for the case of threshold $k=2$, and we also provide results for
%	general $k$ in the paper.
	
In this paper, we show that, unlike simple routing or complex diffusion, in complex routing problem for any $k\geq 2$, for the entire range of $\alpha$, the routing time is polynomial in $n$ both in expectation and with high probability for any decentralized routing algorithm.
Compared with simple routing or complex diffusion, the results at the sweetspot of $\alpha =2$ are the most interesting:
	simple routing has routing time $O(\log^{2} n)$ in expectation \cite{Kleinberg2000small} and complex diffusion has an upper bound of $O(\log^{k+1.5} n)$ in expected diffusion time \cite{Ghasemiesfeh2013complex}, while complex routing
	has a lower bound of $\Omega(n^{\frac{1}{4}})$ in expected routing time, for any $k\geq 2$.
This exponentially wide gap indicates intrinsic difference between complex routing and simple routing or complex diffusion. 
%complex diffusion where all eligible nodes
%	are activated in one step, and complex routing where only one node can be activated at one step.
We further show that if we allow activating $m$ nodes in one step, the routing time is lower bounded
	by $\Omega(n^{\frac{1}{4}}/m)$, which means that
	to get a polylogarithmic upper bound on the routing time $m$ has to be
	$\Omega(n^{\frac{1}{4}}/\log ^c n)$ for some constant $c$.

Our main contribution is that
	we propose the study of complex routing, and prove that the routing time has polynomial
		lower bound in the entire range of $\alpha$ for complex routing.
Our results indicate that
	complex routing is much harder than complex diffusion and the routing time of complex contagion differs exponentially compared to simple contagion at sweetspot.

\subsection{Additional Related Work}

Social and information networks and network diffusions have been extensively studied, and a comprehensive
	coverage has been provided by recent textbooks such as \cite{Easley2010networks,Newman2010networks}.
In this section, we provide most related work in addition to the ones already discussed in the introduction.

Since the proposal of the small-world network models by \cite{Watts1998collective,Kleinberg2000small},
	many extensions and variants have been studied.
For example, Kleinberg proposed a small-world model based on tree structure \cite{Kleinberg2001small},
	Fraigniaud and Giakkoupis extended the model to allow power-law degree distribution
	\cite{FraigniaudG09} or arbitrary base graph structure \cite{FraigniaudG10}.

In terms of network diffusion, a line of research initiated in 
	\cite{Kempe2003maximizing,Kempe2005influential} studied
	the maximization problem of finding a set of small seeds to maximize the influence spread,
	usually under a stochastic diffusion model.
For Chen et al. \cite{Chen2009efficient} provided efficient influence maximization algorithms for large-scale
	networks, while Chen~\cite{Chen2008approximability} proved that minimizing the size of the seed set
	for a given coverage in the fixed threshold model is hard to approximate to any polylogarithmic
	factor.
	
Threshold behavior is also studied in bootstrap percolation \cite{Adler1991bootstrap}, where all
	nodes have the same threshold and initial seeds are randomly selected.
Bootstrap percolation focuses on the study of the critical fraction $f$ of the seed nodes selected so
	that the entire network is infected in the end.
The network structures investigated for bootstrap percolation include grid \cite{Chalupa1979bootstrap},
	trees \cite{Balogh2006bootstrap}, random regular graphs \cite{Balogh2007bootstrap},
	complex networks \cite{Baxter2010bootstrap} etc.

The rest of the paper is organized as follows.
Section~\ref{sec:model} provides the technical model and problem definitions.
Sections~\ref{sec:routing} presents the results and analyses on complex routing.
We conclude the paper in Section~\ref{sec:conclude}.

\section{Model and Problem Definitions} \label{sec:model}
We now provide the precise definitions of the network model,
	the propagation model, and the problems we are studying in this paper.

\subsection{Kleinberg's Small-World Networks}
The Kleinberg's small-world network model defines a random graph based on a set $V$ of $n$ nodes organized in a $\sqrt{n} \times \sqrt{n}$ two-dimensional grid~\cite{Kleinberg2000small}.
For convenience, we connect the top boundary nodes of the grid with the corresponding
	bottom boundary nodes, and
	connect the left boundary nodes with the corresponding right boundary nodes, creating a two-dimensional torus,
	in which the positions of all nodes are symmetric.
%\wei{I changed "identifying two boundaries" to "connecting two boundaries", since
%if we identify them, we will have $(\sqrt{n}-1)^2$ nodes instead of $n$ nodes. Please check
%if this is fine.}
For nodes $u$ and $v$ on the torus, the Manhattan distance $|uv|$ between them is the shortest distance
	from $u$ to $v$ (or $v$ to $u$) using grid edges.
%For node $u$ with grid coordinate $(i,j)$ and node $v$ with coordinate $(k,l)$, the Manhattan distance $|uv|$ between $u$ and $v$ is defined as $|uv| = (k-i \mod \sqrt{n} ) + (l-j|$.

There are two types of edges in this random graph: \textit{strong ties} and \textit{weak ties}.
Strong ties refer to the undirected edges between any pair of nodes with Manhattan distance no more than $p$, where $p \geq 1$ is a universal constant.
%Two nodes connected by a strong tie are {\em local neighbors} to each other.
%\wei{check if need this definition.}
Weak ties refer to random edges connecting any node $u$ with other possibly remote nodes $v$ in the grid.
Each node $u$ has $q$ weak tie connections created independently from one another, and the $i$-th
	weak tie initiated by $u$ has endpoint $v$ with probability proportional to $1/{|uv|}^\alpha$, where $\alpha\geq 0$ is a parameter of the model.
In order to get the probability distribution of weak ties, we multiply $1/{|uv|}^\alpha$ by the normalizing factor $\mathcal{Z} = 1/\sum_{v\in V}|uv|^{-\alpha}$ (on a torus, this value is the same for any $u\in V$).
For a node $u$ in the network, $u$'s {\em grid-neighbors} are nodes linked with $u$ through strong ties while {\em weak-neighbors} are nodes linked with $u$ through weak ties.

The original network model by Kleinberg \cite{Kleinberg2000small} considers the weak tie from $u$ to $v$
	as a directed edge, and we call it the {\em directed Kleinberg's small-world network model},
	while some work
	including \cite{Ghasemiesfeh2013complex} considers the weak ties as undirected edges.
Define random graph $G(n,k,\alpha)$ as directed Kleinberg's small-world network
with $n$ nodes and parameter $\alpha$ and $p=q=k$.
We only consider directed network models in this paper.

%\subsection{complex contagion}
\subsection{Routing of Complex Contagion}

We model the propagation of information, disease, or innovations in a network as a {\em contagion}.
Each node in a network has three possible states --- {\em inactive}, {\em exposed}, {\em infected}
	(or {\em activated}), and a node can transformed from the inactive state to the exposed state and
	then to the infected state, but not in the reverse direction.

A contagion proceeds in discrete time steps $0,1,2,\ldots$.
At time $t\ge 1$, a node becomes exposed if at time $t-1$ at least $k$ of its neighbors (or in-neighbors
	in the case of directed networks) are infected.
An exposed node may become infected immediately or at a later step, which will be specified later.
A {\em simple contagion} refers to the contagion with $k=1$, that is, one infected neighbor is enough to
	expose (and potentially infect) the node,
	while a {\em complex contagion} refers to the case of $k\ge 2$, that is, at least
	two infected neighbors are needed to infect a new node.
We refer the complex contagion with $k\ge 2$ as $k$-complex contagion.

%\subsection{Diffusion of Complex Contagion}
%
%In this paper, we define the {\em diffusion of complex contagion}, or simply {\em complex diffusion},
%	as the contagion process in which all exposed nodes in a time step become infected immediately
%	in the same time step.
%
%To study $k$-complex diffusion, at time $0$, we set $k$ consecutive nodes on the grid in one dimension
%	as infected initially, which we refer as {\em seed nodes}.
%%\wei{Is this a good way to set initial seeds for $k > 2$?}
%For convenience, we also set $p=q=k$.
%When $p=k$, the $k$-complex diffusion is guaranteed to infect all nodes eventually through strong ties only,
%	and $q=k$ provides chances for the $k$-complex contagions to diffuse through weak ties.
%
%We study how fast the contagion diffuses to all nodes in the network, which is determined by the
%	connection of weak ties since strong ties are fixed.
%To do so,  we define \emph{diffusion time} as the number of time steps needed to
%	infect all nodes in the whole network with complex diffusion from the $k$ initial seeds.

We study a different propagation phenomenon closer to the decentralized routing behavior
	studied in \cite{Kleinberg2000small} originally for the small-world network model, which we call
	{\em routing of complex contagion}, or simply {\em complex routing}.
	
To study $k$-complex routing, at time $0$, we set $k$ consecutive nodes on the grid in one dimension
	as infected initially, which we refer as {\em seed nodes}.
For convenience, we also set $p=k$.
When $p=k$, the $k$-complex routing is guaranteed to infect all nodes eventually through strong ties only.
In complex routing, we have a target node $t$ besides the set of $k$ initial seed nodes.

The task is to infect or activate node $t$ as fast as possible.
We can only select one exposed node to activate at each time step.
%while exposed nodes are activated immediately in {\em complex diffusion}.
Moreover, when selecting the node to activate at time $i$, one only knows the
	out-neighbors  of already activated nodes since decentralized routing is applied.
This corresponds to the situation where a group of people try to influence a target by gradually growing
	their allies in the social network towards the target,
	and they only know the friends of their allies and try to recruit
	one of them into the allies at the next time step.
Note that when $k=1$, $k$-complex routing is essentially the decentralized simple routing studied in \cite{Kleinberg2000small}.
 
%if we replace the $k$-complex contagion with simple contagion in the above propagation scheme,
%	and further restrict that the next node to activate is only selected from the neighbors of the
%	newly activated node, then the model is the decentralized routing studied in \cite{Kleinberg2000small}.

To study how fast the routing could be successful, we define the {\em routing time} as the number of
	time steps needed to activate the farthest target node $t$ from the seed node in terms of the
	Manhattan distance.

%\subsection{Statement of Results}
%We give the results of lower bounds for the diffusion time of $2$-complex-contagion in Kleinberg's Small World Model for $0 \leq \alpha <2$ and $\alpha>3$.\\
%Besides, we studied complex routing in Kleinberg's model and give the lower bound for $\alpha = 2$.

\def \doublelinked {\stackrel{ \geq 2}{\longleftrightarrow}}
\def \linked {\stackrel{ \geq 1}{\longleftrightarrow}}
\def \firstto {\stackrel{ 1st}{\longrightarrow}}
\def \secondto {\stackrel{ 2nd}{\longrightarrow}}
\def \doubleto {\stackrel{ \geq 2}{\longrightarrow}}

\section{Results on Complex Routing} \label{sec:routing}
When studying complex routing, we use the directed Kleinberg's small-world network model, same as the model originally proposed by Kleinberg in
	\cite{Kleinberg2000small} for decentralized routing.
As described in the model, we consider decentralized routing in which
	a node can only send activation to its out-neighbors.
Hence only when a node is pointed to by edges from $k$ different activated nodes it becomes exposed. For the strong tie, we still treat them as undirected or bi-directional.
In each time step, we only have the knowledge of the current activated nodes and the out-neighbors of the current activated nodes.
This allows us to apply the Principle of Deferred Decisions \cite{Motwani1995randomized}
	in the same way
	as applied in \cite{Kleinberg2000small}, which means that the weak ties of a node $u$ are defined and known only when $u$ is activated.
Initial seeds set is a set of $k$ consecutive nodes, so the $k$-complex routing will eventually activate target $t$ when we set $p = k$ in Kleinberg's small-world network model.

We consider a $2$-complex routing task from a pair of grid neighbor nodes $S_0 = \{s_0^1, s_0^2\}$  to a destination $t$ where $s_0^1,s_0^2$ have Manhattan distance of $1$ on the grid.
In this paper, we discuss the routing with initial grid distance of $|s_0^1t|=\Theta(\sqrt{n})$.
The strategy of activating nodes from exposed nodes set is not restricted. A special scheme is choosing the node with smallest Manhattan distance to $t$ in each time step, which is the greedy algorithm.
But our result holds for any decentralized node selection schemes, even randomized ones.
The following theorem provides the lower bound result on the routing time.

\begin{theorem}\label{thm:routing}
For any decentralized routing schemes (even randomized ones), the routing time of $2$-complex routing in $G(n,2,\alpha)$ has the following lower bounds based on the parameter $\alpha$, for any small $\varepsilon > 0$:

%with $\alpha = 2$ is $\Omega(n^{1/4})$ with probability at least $1-O(\frac{1}{\log n})$.
\begin{enumerate}
\item
For $\alpha\in [0,2)$, the routing time is $\Omega( n^{\frac{1-\varepsilon}{\alpha + 2}})$ with probability at least $1-O(n^{-\varepsilon})$ and the expected routing time is $\Omega( n^{\frac{1}{\alpha + 2}})$.
\item
For $\alpha = 2$, the routing time is $\Omega(n^{\frac{1}{4}})$ with probability at least $1-O(\frac{1}{\log n})$ and the expected routing time is $\Omega(n^{\frac{1}{4}})$.
\item
For $\alpha\in (2,+\infty)$, the routing time is $\Omega( n^{\frac{\alpha-2\varepsilon}{2(\alpha+2)}})$ with
	probability at least $1-O(n^{-\varepsilon})$ and the expected routing time is $\Omega( n^{\frac{\alpha}{2(\alpha+2)}})$.

\end{enumerate}
\end{theorem}

First we give some necessary definitions.
For a set of nodes $S$, define $\mathcal{E}(S)$ to be the set of exposed nodes for the current activated set $S$, namely $\mathcal{E}(S) = \{x\notin S\ |\ x \mbox{ has at least two in-neighbors in set } S\}$.
%\wei{This definition seems to be wrong! It only considers weak ties according to
%	the definition of symbol $\doubleto$, but a node could be exposed through
%	strong ties too. I think we should just remove this formula and use the
%	English definition, unless I am wrong.}
In a routing protocol, let $S_i$ be the set of the current activated nodes in time $i$. In time step $i$, we can choose at most {\em one} node $u \in \mathcal{E}(S_{i-1})$, and activate $u$ (which means we add $u$ to $S_{i-1}$ in time $i$ and obtain $S_{i}$). From the definition of $\mathcal{E}(S)$ we know that complex routing proceeds following
	 the direction of edges in directed Kleinberg's small-world model.

%In the greedy routing scheme, the node we choose in each round is the node with smallest Manhattan distance from $t$ in $A(S)$.

%In the following subsections, we first analyse the performance of the general routing scheme in directed Kleinberg's small-world model, and then generalize our result to complex routing allowing multiple activation in one time step.

\subsection{Proof of Deterministic Scheme}
We consider deterministic decentralized routing schemes first, and in next subsection we will show
that our lower bounds still hold for randomized schemes.
First we give the proof of routing time for $\alpha = 2$.

Suppose $S_0, S_1,\cdots, S_{\ell}$ is the sequence of activated sets of nodes in routing where $S_i$ is the set of current activated nodes in time step $i$.
The initial seeds are $\{s_0^1,s_0^2\}$ so $S_0 = \{s_0^1,s_0^2\}$, $S_i = \{s_0^1,s_0^2,s_1,\cdots,s_{i}\}$ and in time $i \geq 1$ we add a new node $s_{i}$ selected from $\mathcal{E}(S_{i-1})$, particularly $s_l=t$.
Let $d_i = d(S_i \cup \mathcal{E}(S_i),t)$, where $d(S,u)$ is the minimum Manhattan distance between node $v\in S$ and $u$. It is easy to observe that $d_i$ is a non-increasing sequence and $d_{\ell-1} = d_{\ell}= 0$.
For convenience, we write $s_0^1$ as $S_{-1}$
% so that we can say $s_i$ is newly activated in time step $i$. Also we
and define that $d_{-1} = |s_0^1t| = \sqrt{n}$.
%Since $|s_0^1t| = n^{1/2}$,
We then prove that when the parameter $\alpha = 2$, $\Pr(\forall~0 \leq i < cn^{\frac{1}{4}}, d_{i-1}-d_{i}\leq n^{\frac{1}{4}})$ is high enough, where $c<1$ is a positive constant we will set later. Define event $\chi = \{\forall~0 \leq i < cn^{\frac{1}{4}}, d_{i-1}-d_{i}\leq n^{\frac{1}{4}}\}$. Event $\chi$ means from time step $0$ to $cn^{1/4}-1$, the Manhattan distance between the current activated set and target $t$ decrease at most $n^{\frac{1}{4}}$ in each time step. %At time $0$, we select two grid neighbors as initial seeds.

\begin{lemma}
\label{lem:routing}
For decentralized $2$-complex routing in directed Kleinberg's small-world network $G(n,2,\alpha)$ with $\alpha = 2$, given the initial seeds $\{s_0^1,s_0^2\}$ and farthest target $t$ with $|s_0^1t| = \Theta(\sqrt{n})$,
then for some suitable constant $c \in (0,1)$,
$$\Pr(\forall~0 \leq i < cn^{\frac{1}{4}}, d_{i-1}-d_{i}\leq n^{\frac{1}{4}})\geq 1-O(\frac{1}{\log n}).$$
%holds for some suitable constant $c>0$.
\end{lemma}

\begin{proof}
%For any time step $i \geq 1$, if $d_{i-1}-d_{i}> n^{1/4}$, the gap between $d_{i-1}$ and $d_{i}$ is caused by the adding of node $s_{i}$.
Let $u_{i} = \arg\min_x \{ d(x,t) | x\in S_i \cup \mathcal{E}(S_{i}) \}$,
so $u_{i}$ is the node that is closest to node $t$ and can be activated by set $S_i$ or belong to $S_i$.
%In the greedy routing scheme, $u_{i}$ will be selected in time step $i+1$.
%\wei{Why do we mention greedy routing here? If nothing special, we should remove this sentence.}
Since $u_{i-1}$ is the node that with the shortest Manhattan distance to $t$ among $\mathcal{E}(S_{i-1}) \cup S_{i-1}$
and $s_i \in \mathcal{E}(S_{i-1})$, $|s_{i}t| \geq |u_{i-1}t| = d_{i-1}$.
Thus if $d_{i-1}-d_{i}> 0$, we know that $s_i$ is not the node closest to $t$ among $S_i \cup \mathcal{E}(S_{i})$ since $|s_{i}t| \geq d_{i-1}$.
Besides, we can also get that $u_i \in \mathcal{E}(S_i) \setminus \mathcal{E}(S_{i-1})$
and $s_i$ activate $u_i$ together with another node in $S_{i-1}$.
%the exposing of $u_i$ is the result of the activation of $s_i$.
Combining with the definition that $|u_{i}t| = d_{i}$, we know $|s_{i}u_{i}| \geq |s_{i}t| - |u_{i}t| = d_{i-1}-d_{i}$. Hence we have the following conclusions:

If $d_{i-1}-d_{i}> n^{1/4}$ for $i \geq 1$, then we can conclude
(1) $|s_{i}u_{i}| > n^{\frac{1}{4}}$;
(2) $u_i$ is one of the out-neighbors of $s_i$, more specifically, $s_i$ initiates a weak tie to $u_i$;
(3) $u_i$ is exposed exactly in time step $i$, so there is exactly one weak tie from some node in $S_{i-1}$ to $u_i$.
For $i=0$, because $d(S_0,t) = d_{-1}$, so the gap between $d_{-1}$ and $d_0$ is caused by $u_0 \in \mathcal{E}(S_0)$. The conclusions still hold.

We define the set of nodes that are the endpoints of the weak ties initiated by $S_{i-1}$ as $X_i$.
$X_i$ is indeed the set of weak-neighbors in directed Kleinberg's small-world network.
Apparently $u_i \in X_i$ according to assertion (3) above.
If $d_{i-1}-d_{i}> n^{\frac{1}{4}}$ happens, $u_i$ can be reached by $s_i$ with a weak tie of distance at least $n^{\frac{1}{4}}$.
Define $u \to v$ as node $u$ initiates a weak tie with endpoint $v$.
By union bound, we have:
\begin{equation}
\begin{array}{ll}
& 1- \Pr(\forall~0 \leq i < cn^{\frac{1}{4}}, d_{i-1}-d_{i}\leq n^{\frac{1}{4}})\\
= & \Pr(\exists~0 \leq i < cn^{\frac{1}{4}}, d_{i-1}-d_{i}> n^{\frac{1}{4}}) \\
\leq & \sum_{i=0}^{cn^{\frac{1}{4}}-1}\Pr(d_{i-1}-d_{i}> n^{\frac{1}{4}}) \\
\leq & \sum_{i=0}^{cn^{\frac{1}{4}}-1}\Pr(s_i \to u_i , |s_iu_i| > n^{\frac{1}{4}} , u_i \in X_i) \\
\leq & \sum_{i=0}^{cn^{\frac{1}{4}}-1}\Pr(\exists~x \in X_i, s_i \to x, |s_ix| > n^{\frac{1}{4}}).
\end{array}
\label{eq:routing}
\end{equation}

Since there is $i+1$ nodes in the set $S_{i-1}$,
$S_{i-1}$ initiate $q(i+1)$ weak ties, which means that $|X_i| \leq q(i+1)$.
Denote $\mathcal{H}_{i}\subseteq 2^{V}$ to be the set of all sets of nodes with size no more than $q(i+1)$. Then we fix the randomness of $X_i$ and $s_i$:
\begin{equation*}
\begin{array}{ll}
& \Pr(\exists~x \in X_i, s_i \to x, |s_ix| > n^{\frac{1}{4}})\\
\leq & \sum_{C \in \mathcal{H}_{i}} \sum_{v \in V} \Pr \big( (X_i=C) \land  (s_i=v)
\land (\exists~x \in C, v \to x, |vx| > n^{\frac{1}{4}}) \big)\\
= & \sum_{C \in \mathcal{H}_{i}} \sum_{v \in V} \Pr \big( (X_i=C) \land  (s_i=v) \big)
\Pr(\exists~x \in C, v \to x, |vx| > n^{\frac{1}{4}})\\
\leq & \sum_{C \in \mathcal{H}_{i}} \sum_{v \in V} \Pr \big( (X_i=C) \land  (s_i=v) \big)
\sum_{x \in C} \Pr(v \to x, |vx| > n^{\frac{1}{4}})\\
\leq & \sum_{C \in \mathcal{H}_{i}} \sum_{v \in V} \Pr \big( (X_i=C) \land  (s_i=v) \big) \cdot
|C| \cdot 2\mathcal{Z}\frac{1}{n^{2\cdot 1/4}}\\
\leq &  q(i+1) \cdot 2\frac{\mathcal{Z}}{ n^{1/2}} \sum_{C \in \mathcal{H}_{i}} \sum_{v \in V}
\Pr \big( (X_i=C) \land  (s_i=v) \big)
=  2q(i+1) \cdot \frac{\mathcal{Z}}{ n^{1/2}}.
\end{array}
\end{equation*}
By the property of decentralized routing, event $\{(X_i=C) \land (s_i = v)\}$ only depends on the random set $S_{i-1}$ and the outgoing weak ties from $S_{i-1}$, and $v$ is not in $S_{i-1}$, while event $\{\exists x \in C, v \to x, |vx| > n^{\frac{1}{4}} \}$ only depends on the outgoing weak ties of the fixed node $v$. Thus event $\{(X_i=C) \land (s_i = v) \}$ is independent of event $\{ \exists x \in C, v \to x, |vx| > n^{\frac{1}{4}} \}$.
This gives us the first ``$=$" in the equation.
For a node $x$, if $|vx| \leq n^{\frac{1}{4}}$, then $\Pr(v \to x, |vx| > n^{\frac{1}{4}}) = 0$; otherwise, $\Pr(v \to x, |vx| > n^{\frac{1}{4}}) \leq 2p(v,x) \leq 2\mathcal{Z}\frac{1}{ n^{2\cdot 1/4}}$. Hence we have the third ``$\leq$".
Substitute it into Inequality~(\ref{eq:routing}):
\begin{equation*}
\begin{array}{ll}
& 1- \Pr(\forall~0 \leq i < cn^{\frac{1}{4}}, d_{i-1}-d_{i}\leq n^{\frac{1}{4}})\\
\leq & \sum_{i=0}^{cn^{\frac{1}{4}}-1}\Pr(\exists~x \in X_i, s_i \to x, |s_ix| > n^{\frac{1}{4}})\\
\leq & \sum_{i=0}^{cn^{\frac{1}{4}}-1} q(i+1) \cdot 2\frac{\mathcal{Z}}{ n^{1/2}} \\
\leq & cn^{\frac{1}{4}} \cdot qcn^{\frac{1}{4}} \cdot \Theta(\frac{1}{\log{n}})\frac{2}{ n^{1/2}} =  O(\frac{1}{\log n}).\\
\end{array}
\end{equation*}
\end{proof}

Due to the above lemma, it is easy to see that the routing time is at least $cn^{\frac{1}{4}}$ with high probability for $\alpha = 2$.
\begin{proof}[of Theorem~\ref{thm:routing} (deterministic routing scheme)]
Lemma~\ref{lem:routing} says, for $\alpha = 2$, in the first $cn^{\frac{1}{4}}$ steps, the grid distance between the current activated set and target $t$ decreases at most $n^{\frac{1}{4}}$ in each step. Thus, for the first $cn^{\frac{1}{4}}$ steps, target $t$ does not belong to the activated set and the routing procedure will continue.
Hence with probability of $1-O(\frac{1}{\log n})$, to activate the target $t$ in $G(n,2,\alpha)$ with $\alpha=2$, decentralized $2$-complex routing needs at least $cn^{\frac{1}{4}}$ time steps. The expected routing time is $cn^{\frac{1}{4}} \cdot (1-O(\frac{1}{\log n})) = \Omega(n^{\frac{1}{4}})$.

When $\alpha\in [0,2)$, like the proof in Lemma~\ref{lem:routing}, we can prove for small $\varepsilon > 0$,
\begin{equation*}
\begin{array}{ll}
& 1- \Pr(\forall~0 \leq i < cn^{\frac{1-\varepsilon}{\alpha+2}},
d_{i-1}-d_{i}\leq 2n^{\frac{\alpha + 2\varepsilon}{2(\alpha+2)}})\\
\leq & \sum_{i=0}^{cn^{\frac{1-\varepsilon}{\alpha+2}}-1}\Pr(\exists~x \in X_i, s_i \to x,
|s_ix| > 2n^{\frac{\alpha + 2\varepsilon}{2(\alpha+2)}})\\
\leq & \sum_{i=0}^{cn^{\frac{1-\varepsilon}{\alpha+2}}-1} 2(i+1) \cdot \mathcal{Z} \cdot 2(2n^{\frac{\alpha + 2\varepsilon}{2(\alpha+2)}})^{-\alpha} \\
\leq & cn^{\frac{1-\varepsilon}{\alpha+2}}
\cdot 2cn^{\frac{1-\varepsilon}{\alpha+2}} \cdot \Theta(\frac{1}{n^{1-\alpha/2}})\cdot
O(n^{-\frac{\alpha(\alpha + 2\varepsilon)}{2(\alpha+2)}})
=  O(n^{-\varepsilon}).\\
\end{array}
\end{equation*}

So the routing time is $\Omega(n^{\frac{1-\varepsilon}{\alpha+2}})$ with probability at least $1-O(n^{-\varepsilon})$. By setting $\varepsilon = 0$ and adjusting the parameter $c$, the expected routing time can be obtained.

When $ \alpha >2$, we can prove that for small $\varepsilon > 0$,
$$\Pr(\forall~0 \leq i < cn^{\frac{\alpha-2\varepsilon}{2(\alpha+2)}},
d_{i-1}-d_{i}\leq n^{\frac{1 + \varepsilon}{\alpha+2}}) \geq 1-O(n^{-\varepsilon})$$ like the proof above. Hence with probability at least $1-O(n^{-\varepsilon})$, we need $cn^{\frac{\alpha-2\varepsilon}{2(\alpha+2)}}$ time steps to find the target. %Let $\varepsilon = 0$, it is not difficult to get the bound for the expected routing time.
Similarly we can get the bound for the expectation.
\end{proof}

\subsection{Proof of Randomized Scheme}
The lower bound of routing time still holds for randomized decentralized routing scheme. Here we just provide a brief proof for the case when $\alpha = 2$.
\begin{proof}[of Theorem~\ref{thm:routing} (randomized routing scheme)]

Let $\mathcal{A}$ be the set of all the deterministic decentralized routing schemes, $\mathcal{G}$ be the set of all the possible networks based on Kleinberg's small-world model with $\alpha = 2$ and $\pi$ be the corresponding distributions over $\mathcal{G}$. Further define $T(A,G')$ as the routing time of an algorithm $A \in \mathcal{A}$ applied on a network $G'\in \mathcal{G}$.

Since a randomized algorithm is just a distribution over all deterministic algorithms, all we need to show is that for any distribution $\mu$ over $\mathcal{A}$,
\begin{equation}\label{eq:randomized_alg1}
\Pr_{G'\sim \pi}(\mathbb{E}_{A\sim \mu}(T(A,G'))=\Omega(n^{1/4}))=1-O(\frac{1}{\log{n}}),
\end{equation}
which means that for any randomized scheme $A$, with high probability (w.r.t the distributions over $\mathcal{G}$), the expected routing time of $A$  is no less than $\Omega(n^{1/4})$. That is, any randomized scheme is slow for most of the input small-world networks.

From the previous proof we know any deterministic algorithms
$A\in \mathcal{A}$,
$\Pr_{G'\sim \pi}(T(A,G')=\Omega(n^{1/4}))=1-O(\frac{1}{\log{n}})$.
Thus there exists constant $c_1,c_2>0$ not depending on the scheme $A$, such that
$
\forall~A\in \mathcal{A},~\Pr_{G'\sim \pi}(T(A,G')\geq c_1n^{1/4})\geq 1-\frac{c_2}{\log{n}}.
$
Since it holds for all $A\in \mathcal{A}$, it is also true that
\begin{equation}\label{eq:randomized_alg3}
\Pr_{G'\sim \pi\atop A\sim \mu}(T(A,G')\geq c_1n^{1/4})\geq 1-\frac{c_2}{\log{n}}.
\end{equation}
We claim that
\begin{equation}\label{eq:randomized_alg2}
\Pr_{G'\sim \pi}(\mathbb{E}_{A\sim \mu}(T(A,G'))\geq \frac{c_1}{2}n^{1/4})\geq 1-\frac{3c_2}{\log{n}},
\end{equation}
which is a quantitive version of Eq.~(\ref{eq:randomized_alg1}). Suppose Inequality~(\ref{eq:randomized_alg2}) is false, then
$$
\Pr_{G'\sim \pi}(\mathbb{E}_{A\sim \mu}(T(A,G'))< \frac{c_1}{2}n^{1/4}) > \frac{3c_2}{\log{n}}.
$$
Let $\mathcal{G}_1\subseteq \mathcal{G}$ be the set of all the $G'$ satisfies $\mathbb{E}_{A\sim \mu}(T(A,G'))< \frac{c_1}{2}n^{1/4}$,
then the above inequality tells us that $\Pr_{G'\sim \pi}(G'\in \mathcal{G}_1) > \frac{3c_2}{\log{n}}$.
For any fixed $G'\in \mathcal{G}_1$, by Markov Inequality
$$\Pr_{A\sim\mu}(T(A,G')\geq c_1n^{1/4})\leq \frac{\mathbb{E}_{A\sim\mu}(T(A,G'))}{c_1n^{1/4}}<\frac{1}{2}.$$
So $\Pr_{A\sim\mu}(T(A,G')< c_1n^{1/4})\geq 1/2$ for any $G'\in \mathcal{G}_1$. Therefore,
$$\Pr_{A\sim\mu\atop G'\sim \pi}(T(A,G')< c_1n^{1/4}) > \frac{1}{2}\cdot \frac{3c_2}{\log{n}}=\frac{1.5c_2}{\log{n}},$$
which contradict to Inequality~(\ref{eq:randomized_alg3}).
\end{proof}

\subsection{Discussion and Extension}
We describe complex routing as the task of activating a node as fast as possible.
Here we consider the task of activating a target node $t$ that is $n^{\gamma}$ grid distance away from the seeds.
Similar hardness results can be inferred if we determine the step size and step number of routing cautiously.
Here we just sate the theorem and do not provide the redundant proof.
\begin{theorem}\label{thm:routing_dis}
For any decentralized routing schemes (even randomized ones), the routing time for $2$-complex routing to activate a node with Manhattan distance $n^{\gamma} (0<\gamma\leq \frac{1}{2})$ away in $G(n,2,\alpha)$ has the following lower bounds based on the parameter $\alpha$, for any small $\varepsilon > 0$:
\begin{enumerate}
\item
For $\alpha\in [0,2)$, the routing time is $\Omega( n^{\frac{2\gamma(1-\varepsilon)}{\alpha + 2}})$ with probability at least $1-O(n^{-\varepsilon'})$ where $\varepsilon' = \max{\{ 2\gamma\varepsilon, (1-2\gamma)(1-\alpha/2)\}}$ and the expected routing time is $\Omega( n^{\frac{2\gamma}{\alpha + 2}})$.
\item
For $\alpha = 2$, the routing time is $\Omega(n^{\frac{\gamma}{2}})$ with probability at least $1-O(\frac{1}{\log n})$ and the expected routing time is $\Omega(n^{\frac{\gamma}{2}})$.
\item
For $\alpha\in (2,+\infty)$, the routing time is $\Omega( n^{\frac{\gamma(\alpha-2\varepsilon)}{\alpha+2}})$ with
	probability at least $1-O(n^{-\varepsilon})$ and the expected routing time is $\Omega( n^{\frac{\gamma\alpha}{\alpha+2}})$.

\end{enumerate}
\end{theorem}

We can obtain the same lower bound of routing time for $k$-complex routing. To ensure the success of complex routing, let $p = q = k$ for the Kleinberg's small-world network model and the size of seed nodes is $k$. The result is the same with $2$-complex routing so we omit it.

%complex routing allowing multiple activation in one time step
Next, we extend our results to complex routing where at most $m$ nodes can be activated in each time step.
When $m=1$, the result is what we covered in Theorem~\ref{thm:routing} for
	complex routing.
When we do not restrict $m$, complex routing becomes complex diffusion.
Thus a general $m$ allows us to connect complex routing with diffusion, and see how
	large $m$ is needed to bring down the polynomial lower bound in complex routing.
From the theorem we know that we would not get polylogarithmic routing time for complex routing in $G(n,2,2)$ where $m$ nodes can be activated in each step, unless 
$m = n^{\frac{1}{4}} / \log^{O(1)} n$.

%As to the complex routing discussed above, we just choose one exposed node and activate it in each step. We get polynomial lower bound for complex routing in directed Kleinberg's small-world network model with $\alpha = 2$.
%We want to study the complex routing with multiple activation. If we don't restrict the number of nodes that transform from the state of exposed to activated in each step, the complex routing become complex diffusion indeed. Now we make a trade-off in the number of nodes activated in each step, that is between complex routing and complex diffusion. Assuming that at most $m$ nodes can be activated in each time step, then we see if with multiple activation we can get ploylogarithmic lower bound.

\begin{theorem}
In decentralized routing, for $k$-complex routing in $G(n,2,\alpha)$, if at most $m$ nodes can be activated in each time step, routing time 
has the following lower bounds based on the parameter $\alpha$, for any small $\varepsilon > 0$:
\begin{enumerate}
\item
For $\alpha\in [0,2)$, the routing time is $\Omega( n^{\frac{1-\varepsilon}{\alpha + 2}}/m)$ with probability at least $1-O(n^{-\varepsilon})$ and the expected routing time is $\Omega( n^{\frac{1}{\alpha + 2}}/m)$.
\item
For $\alpha = 2$, the routing time is $\Omega(n^{\frac{1}{4}}/m)$ with probability at least $1-O(\frac{1}{\log n})$ and the expected routing time is $\Omega(n^{\frac{1}{4}}/m)$.
\item
For $\alpha\in (2,+\infty)$, the routing time is $\Omega( n^{\frac{\alpha-2\varepsilon}{2(\alpha+2)}}/m)$ with
	probability at least $1-O(n^{-\varepsilon})$ and the expected routing time is $\Omega( n^{\frac{\alpha}{2(\alpha+2)}}/m)$.

\end{enumerate}
\end{theorem}

\begin{proof}
Assuming that $S$ is the set of current activated nodes.
In next time step, we can activate $m$ nodes with the knowledge of the out-neighbors of $S$.
But consider the original complex routing, we just activate one node in each step and we have $m$ time steps to activate nodes.
After each small step, we have the knowledge of the newly added node. Hence the method of activating $m$ nodes with $m$ time steps is more effective than infecting $m$ nodes in just one time step.
Therefore if we need $T$ time steps to find the target with original complex-routing, the routing time with activating $m$ nodes in each time step is at least $\frac{T}{m}$.
The expected routing time is $\frac{T}{m} \cdot (1-O(\frac{1}{\log n}))$.
Then the theorem follows.
\end{proof}

\section{Conclusion} \label{sec:conclude}

In this paper, we study the routing of complex contagion in Kleinberg's small-world networks.
We show that for complex routing the routing time is lower bounded by a polynomial in the number of nodes in the network
	for the entire range of $\alpha$, which is qualitatively different from the polylogarithmic
	upper bound in both complex diffusion
	and simple routing for $\alpha = 2$.
Our results indicate that
	complex routing is much harder than both complex diffusion and
	simple routing at the sweetspot.

There are a number of future directions of this work.
One may look into complex routing for undirected small-world networks or other variants of
	the small-world models.
The qualitative difference between complex diffusion and complex routing for the case of $\alpha = 2$
	may worth further investigation.
For example, one may study if there is similar difference for a larger class of graphs, and under
	what network condition complex routing permits polylogarithmic solutions.

%

%\begin{thebibliography}{}  % (do not forget {})
%\bibitem[1982]{clar:eke}
%Clarke, F., Ekeland, I.:
%Nonlinear oscillations and boundary-value problems for
%Hamiltonian systems.
%Arch. Rat. Mech. Anal. 78, 315--333 (1982)
%\end{thebibliography}

\begingroup
\bibliographystyle{plainnat}
\bibliography{complexRouting}
\endgroup

\end{document}